\documentclass[a4,12pt]{article}

\usepackage{algorithm,algpseudocode,amsmath,amssymb,array,ascmac,bm,color,colortbl,fancyhdr,geometry,graphicx,here,listings,lscape,makeidx,mathptm,multirow,subfigure,tabularx,theorem,times,txfonts,url,verbatim}

\let\oldenumerate\enumerate
\renewcommand{\enumerate}{
   \oldenumerate
   \setlength{\itemsep}{1pt}
   \setlength{\parskip}{0pt}
   \setlength{\parsep}{0pt}
}

\makeatletter
\long\def\@makecaption#1#2{
  \vskip\abovecaptionskip
  \iftdir\sbox\@tempboxa{#1\hskip1zw#2}
    \else\sbox\@tempboxa{#1.~ #2}
  \fi
  \ifdim \wd\@tempboxa >\hsize
    \iftdir #1\hskip1zw#2\relax\par
      \else #1.~ #2\relax\par\fi
  \else
    \global \@minipagefalse
    \hbox to\hsize{\hfil\box\@tempboxa\hfil}
  \fi
  \vskip\belowcaptionskip}
\makeatother

{\theorembodyfont{\normalfont}
\newtheorem{definition}{Definition}

}

\newtheorem{proposition}{Proposition}
\newtheorem{lemma}{Lemma}

\newtheorem{corollary}{Corollary}

\newcommand{\qed}{\hfill$\blacksquare$\par}
\newenvironment{proof}[1][]
{\ifx #1%
\textbf{Proof. }%
\else%
\textbf{Proof#1. }\fi}
{\qed}

\usepackage[square,comma,numbers,sort&compress]{natbib}

\newcommand{\asconv}{\stackrel{\mathrm{\footnotesize a.s.}}{\longrightarrow}}

\newcommand{\E}{\mathrm{E}}

\newcommand{\md}{\mathrm{d}}

\title{\Large \textbf{Evolutionary model of a population of DNA sequences\\through the interaction with an environment and\\its application to speciation analysis}}

\author{
Hitoshi Koyano${}^{1 \ast}$ and Kouji Yano${}^{2}$\\
\normalsize{${}^{1}$Laboratory for Physical Biology, Quantitative Biology Center,}\\
\normalsize{Institute for Physical and Chemical Research (RIKEN),}\\
\normalsize{2-2-3 Minatojima-Minamimachi, Chuo-ku, Kobe, Hyogo 650-0047, Japan}\\
\normalsize{${}^{2}$Department of Mathematics and Mathematical Analysis,}\\
\normalsize{Graduate School of Science, Kyoto University,}\\
\normalsize{Kitashirakawa-Oiwakecho, Sakyo-ku, Kyoto 606-8502, Japan}\\
\normalsize{$^\ast$Corresponding author. E-mail: hitoshi.koyano@riken.jp}
}

\date{}

\makeindex

\begin{document}

\maketitle

\begin{abstract}
In this study, we construct an evolutionary model of a population of
DNA sequences interacting with the surrounding environment on the
topological monoid $A^{\ast}$ of strings on the alphabet $A = \{
\mathtt{a}, \mathtt{c}, \mathtt{g}, \mathtt{t} \}$.
A partial differential equation governing the evolution of the DNA
population is derived as a kind of diffusion equation on $A^{\ast}$.
Analyzing the constructed model in a theoretical manner, we present
conditions for sympatric speciation, the possibility of which
continues to be discussed.
It is shown that under other same conditions one condition determines
whether sympatric speciation occurs or the DNA population continues to
move around randomly in a subset of $A^{\ast}$.
We next demonstrate that the population maintains a kind of equilibrium
state under certain conditions.
In this situation, the population remains nearly unchanged and does not
differentiate even if it can differentiate into others.
Furthermore, we calculate the probability of sympatric speciation and
the time expected to elapse before it.
\end{abstract}

\noindent {\bf Keywords:} Topological monoid of strings, probability
theory, population of DNA sequences, evolution, speciation.

\section{Introduction}
\label{section:Introduction}

Evolution of a biological population is, at the most fundamental
level, temporal change in a set of DNA sequences that the population
has.
Individuals, therefore DNA sequences that they have, living in an
environment are under selection pressure from the environment and
leave offspring the number of which is determined according to the
pressure operating on them.
Offspring's DNA sequences can contain mutations that occur randomly.
Consequently, a population of DNA sequences gradually changes as
generations pass.

Let us attempt to formulate the sentences in the previous paragraph as
a mathematical model as naturally and faithfully to them as possible.
We denote the set of strings on an alphabet $A = \{ \mathtt{a},
\mathtt{c}, \mathtt{g}, \mathtt{t} \}$ (i.e., finite sequences of
elements of $A$) by $A^{\ast}$.
Any DNA sequence is represented as an element of $A^{\ast}$.
Let $A^{\ast}$ form a metric space, provided with the Levenshtein
distance (denoted by $d_{L}$ hereafter), although various distance
functions such as the extended Hamming distance, longest common
subsequence distance, and Damerau--Levenshtein distance are defined on
$A^{\ast}$, depending on the problem to be considered.
The Levenshtein distance between two strings is the minimum number of
three types of edit operation, insertion, deletion, and substitution,
necessary to transform one into another and is frequently used to
evaluate evolutionary change in DNAs and gene sequences.
In this setting, temporal change in a set of DNA sequences that a
biological population has can be captured as change in a subset of
$A^{\ast}$ with time.
Each DNA sequence in an environment is under selection pressure
exerted by the environment and the pressure is determined according to
the order of nucleotides that compose the sequence, in other words,
its position in $A^{\ast}$.
Therefore, the selection pressure is a nonnegative real-valued
function on $A^{\ast}$ defined when an environment is given.
The selection pressure can be considered a counterpart in evolutionary
biology of a force field in physics, although it is not a
vector-valued function on a vector space.
Furthermore, mutations randomly contained in offspring sequences are
captured as a random transformation on $A^{\ast}$ that does not always
return the same element of $A^{\ast}$ for each element of $A^{\ast}$.
In the following sections, advancing this viewpoint, we develop an
evolutionary model of a biological population as the dynamics of a
subset of $A^{\ast}$ by combining functions and operators on
$A^{\ast}$.

Various evolutionary models of biological populations have been
proposed.
See, for example,
\cite{Templeton_1980,Higgs_1991,Manzo_1994,Geritz_1997,Geritz_1998,Doebeli_2000,Frossel_2000}
for articles that put emphasis on the mechanisms of speciation.
However, any model was constructed on the set $\mathbb{R}$ of real
numbers or real vector space $\mathbb{R}^{n}$.
No studies have been performed that formulated the evolution of a
biological population as a mathematical model at the fundamental level
of temporal change in a set of DNA sequences that the population has,
theoretically analyzed the model, and performed numerical experiments
based on the model.
The reason for this would be that no tools have been prepared for
model construction and analysis because $A^{\ast}$ has not been a
space on which probability theory and statistics should be
constructed, although the phenomenon that a population of DNA
sequences evolves need to be treated in a probabilistic manner because
it is necessary to incorporate mutations that occur randomly in a
model.

In various dynamical systems in which lots of particles interact in a
numerical vector space, particles that is analysis objects are
represented as their coordinates, i.e., numerical vectors.
By contrast, an evolutionary process can be seen as temporal change in
a subset of $A^{\ast}$ because our objects, DNA sequences, are
represented as elements of $A^{\ast}$.
In this study, we model the phenomenon described in the first
paragraph, adopting the course described in the second paragraph based
on probability theory on $A^{\ast}$ developed in recent years
\cite{Koyano_2010,Koyano_2014,Koyano_2016,Koyano_submitted}.
We subsequently investigate the process of speciation (see, for
example, \cite{Coyne_2004,Nosil_2012,Dieckmann_2004} for a review)
that is one of the most fundamental phenomena in biology based on the
model.
Here we construct an evolutionary model for asexual
populations, but it is not difficult to extend the model
into one for sexual populations.

\section{Formulation of the model}
\label{sec:Formulation_of_the_model}

\subsection{Evolutionary model of a DNA population}
\label{subsec:Evolutionary_model_of_a_DNA_population}

In this subsection, we formulate a model of the evolution of a
population of DNA sequences interacting with a surrounding environment
on the topological monoid $A^{\ast}$.
We denote a population of DNA sequences at time $t \in [0, \infty)$ by
$S(t)$.
If $s \in S(t)$, we have $s \in A^{\ast}$.
Let $|X|$ represent the number of elements of $X$ if $X$ is a set and
$|s|$ represent the length of $s$ (i.e., the number of elements of $A$
that composes $s$) if $s \in A^{\ast}$.
We set $n(t) = |S(t)|$ and write $S(t) = \{ s_{1}(t), \cdots,
s_{n(t)}(t) \}$.
Let $R(t)$ be a complete system of representatives of the quotient set
$S(t)/\!=$ of $S(t)$ with respect to the equality relation $=$.
We denote the number of elements of $S(t)$ that are equal to $s
\in R(t)$ (i.e., the number of elements of the equivalent class
of $s \in R(t)$) by $v(s, t)$ and set $q(s, t) = v(s, t) / n(t)$.
We define $Q(t) = \{ q(s, t): s \in R(t) \}$ and refer to $Q(t)$ as
the relative frequency distribution of $S(t)$.
In the following, we suppose that the population size $n(t)$ is
exogenously given.

Let $\mathcal{E}$ represent the set of possible environments in
geographical regions of the earth.
We denote the selection pressure exerted by an environment $E \in
\mathcal{E}$ on a DNA sequence $s \in A^{\ast}$ in the case where $s$
exists in $E$ by $p(s, E)$.
Therefore, the selection pressure is a function $p( \ \cdot \ ,
E) : A^{\ast} \to [0,\infty)$ defined when $E \in \mathcal{E}$ is given.
As formulated below, as $p(s, E)$ is larger, the fitness of $s$ to $E$
is lower, and consequently, the number of offspring of $s$ is smaller,
and vice versa.
We suppose that there exists a level of selection pressure under which
any $s \in A^{\ast}$ cannot produce offspring in $E$ and denote its
critical value by $p_{E}$.
In Subsection \ref{subsec:Formulation_of_the_model_continuation}, we
introduce a model of $p(s, E)$.

We suppose that a ratio $\gamma$ of sequences in $S(t)$ produce
offspring and die in a unit time $[t, t + \Delta t]$ for a constant
$\gamma \in (0, 1]$.
In this setting, the life span varies depending on the sequence if
$\gamma < 1$, whereas the generation overlapping between a parent and
children is not considered.
We denote the set of sequences in $S(t)$ that produce offspring and
die in $[t, t + \Delta t]$ by $\hat{S}(t)$ and
suppose that $\hat{S}(t)$ is determined according to the following
rule C$_{1}$:
For a minimum integer $\hat{n}(t)$ satisfying $\hat{n}(t) \geq \gamma
n(t)$,
$\hat{S}(t)$ is a set of the first $\hat{n}(t)$ sequences when
sequences in $S(t)$ are sorted in descending order with respect to
the living time, where the order of sequences that are equal in
living time is arbitrary.
Let $g_{t}$ represent a mapping that returns $\hat{S}(t) \subset S(t)$
determined according to the rule C$_{1}$ given $S(t)$, in other words
$\hat{S}(t) = g_{t}(S(t))$.

Let $\hat{R}(t)$ be a complete system of representatives of the
quotient set $\hat{S}(t)/\!=$ of $\hat{S}(t)$ with respect to the
equality relation $=$.
We denote the number of elements of $\hat{S}(t)$ that are equal to $s
\in \hat{R}(t)$ by $\hat{v}(s, t)$.
We define the function $f( \ \cdot \ , t, E): \hat{R}(t) \to [0, 1]$
as
\[
f(s, t, E) = \left\{
\begin{array}{ll}
\displaystyle{\frac{p_{E} - p(s, E)}{\displaystyle{\sum_{s^{\prime} \in
\hat{R}(t)}} (p_{E} - p(s^{\prime}, E))}} & \mbox{if } p(s, E) < p_{E},\\
0 & \mbox{otherwise}.
\end{array}
\right.
\]
We have $0 \leq f(s, t, E) \leq 1$ for any $s \in \hat{R}(t)$ and
$\sum_{s \in \hat{R}(t)} f(s, t, E) = 1$.

We suppose that the total number $o(s, t)$ of offspring of sequences
in $\hat{S}(t)$ that are equal to $s$ for each $s \in \hat{R}(t)$ is
determined according to the following rule C$_{2}$:
Letting $\lfloor x \rfloor$ represent the integer part of $x \in
\mathbb{R}$, we set
\[
\tilde{n}(t) = n(t + \Delta t) - n(t) + \hat{n}(t), \quad
u(t) = \tilde{n}(t) - \sum_{s \in \hat{R}(t)} \lfloor \tilde{n}(t)
f(s, t, E) \rfloor.
\]
$\tilde{n}(t)$ represents the number of offspring that $\hat{S}(t)$
can produce.
We sort sequences in $\hat{R}(t)$ in descending order with respect to
$\tilde{n}(t) f(s, t, E)$.
The order of sequences that have an equal value of $\tilde{n}(t) f(s,
t, E)$ is arbitrary.
Then, $o(s, t) = \lfloor \tilde{n}(t) f(s, t, E) \rfloor + 1$ for the
first to $u(t)$th sequence $s$ in $\hat{R}(t)$ and $o(s, t) = \lfloor
\tilde{n}(t) f(s, t, E) \rfloor$ for the $(u(t) + 1)$th to last
sequence $s$ in $\hat{R}(t)$.
This is a rule of determining the number of offspring that
satisfies $\tilde{n}(t) = \sum_{s \in \hat{R}(t)} o(s , t)$ and is as
faithful as to $p(s, E)$ possible.
We set $o(s, t) = 0$ for $s \notin \hat{R}(t)$.
We denote the set of offspring sequences of $\hat{S}(t)$ by
$\tilde{S}(t)$.
$f(s, t, E)/\hat{v}(s, t)$ and $o(s , t)/\hat{v}(s, t)$ can be
regarded as the relative fitness and the fitness of one of sequences
in $\hat{S}(t)$ that are equal to $s \in \hat{R}(t)$, respectively.

Let $m$ be a random transformation on $A^{\ast}$ that does not
necessarily maps each element of $A^{\ast}$ to the same element of
$A^{\ast}$.
We denote the probability that $m$ outputs $s^{\prime}$ given $s$ as
an input by $\mu(s, s^{\prime})$ for $s, s^{\prime} \in A^{\ast}$.
We refer to $m$ as a mutation operator if $m$ satisfies the following
three conditions for $s, s^{\prime}, s^{\prime \prime} \in A^{\ast}$:
(i) $d_{L}(s, s^{\prime}) < d_{L}(s, s^{\prime \prime})$ implies
$\mu(s, s^{\prime}) > \mu(s, s^{\prime \prime})$, (ii) $d_{L}(s,
s^{\prime}) = d_{L}(s, s^{\prime \prime})$ implies $\mu(s, s^{\prime})
= \mu(s, s^{\prime \prime})$, and (iii) mutated sequences generated by
$m$ are independent.
In Subsection~\ref{subsec:Formulation_of_the_model_continuation},
we construct $m$ that satisfies these conditions for $d_{L}$ in a
concrete manner, using the operation of concatenation defined on
$A^{\ast}$, which makes $A^{\ast}$ form a monoid.
We define the mutation operator at time $t$ as a set $m_{t} =
(m_{t}^{(1)}, \cdots, m_{t}^{(\tilde{n}(t))})$ of $\tilde{n}(t)$
mutation operators that operates a set $\{ s_{1}, \cdots,
s_{\tilde{n}(t)} \}$ of $\tilde{n}(t)$ strings in the following
component-wise manner:
\[
m_{t}(s_{1}, \cdots, s_{\tilde{n}(t)}) = (m_{t}^{(1)}(s_{1}), \cdots,
m_{t}^{(\tilde{n}(t))}(s_{\tilde{n}(t)})).
\]
We denote the probability that an offspring sequence contains a
mutation corresponding to Levenshtein distance one by $\pi \in (0,
1)$.

For any $n \in \mathbb{N}$ ($\mathbb{N}$ represents the set of natural
numbers including zero) and $s \in A^{\ast}$, we set the symbol
$\bigcup^{n} \{ s \}$ ($\bigcup$ has only a superscript) as
$\bigcup^{n} \{ s \} = \emptyset$ for $n = 0$ and $\bigcup^{n} \{ s \}
=$ a set of $n$ $s$s for $n \geq 1$.
We define the replication operator $r_{t} = (r_{t}^{(1)}, \cdots,
r_{t}^{(\hat{n}(t))}): (A^{\ast})^{\hat{n}(t)} \to
(A^{\ast})^{\tilde{n}(t)}$ at time $t$ as
\[
r_{t}(s_{1}, \cdots, s_{\hat{n}(t)}) = ( r_{t}^{(1)}(s_{1}), \cdots,
r_{t}^{(\hat{n}(t))}(s_{\hat{n}(t)}) ) = \bigcup_{i = 1}^{\hat{n}(t)}
\bigcup^{o(s_{i}, t)} \{ s_{i} \}.
\]

Lastly, we define the generating operator $G_{t}$ at time $t$ as
$G_{t}(S(t)) = S(t) \smallsetminus g_{t}(S(t)) \cup m_{t} \circ r_{t}
\circ g_{t}(S(t)) $ ($\circ$ represents composition).
A population $S(t + \Delta t)$ is generated from a population $S(t)$
in the following manner:
\begin{eqnarray*}
m_{t} \circ r_{t} \circ g_{t}(S(t))
&=& m_{t} \circ r_{t}(\hat{S}(t))
 \ = \ m_{t} \bigl( r_{t}^{(1)} (\hat{s}_{1}(t)), \cdots, r_{t}^{(\hat{n}(t))}
(\hat{s}_{\hat{n}(t)}(t)) \bigr)\\
&=& m_{t} \left( \bigcup_{i = 1}^{\hat{n}(t)} \bigcup^{o(\hat{s}_{i}(t), t)}
\{ \hat{s}_{i}(t) \} \right)
 \ = \ \bigl\{ m_{t}^{(1)}(s_{1}^{\prime}), \cdots, m_{t}^{(\tilde{n}(t))}
(s_{\tilde{n}(t)}^{\prime}) \bigr\}\\
 &=& \{ \tilde{s}_{1}(t), \cdots, \tilde{s}_{\tilde{n}(t)}(t) \}
 \ = \ \tilde{S}(t),\\
G_{t}(S(t))
&=& S(t) \smallsetminus g_{t}(S(t)) \cup m_{t} \circ r_{t} \circ
g_{t}(S(t))
 \ = S(t) \smallsetminus \hat{S}(t) \cup \tilde{S}(t)
 \ = \ S(t + \Delta t),
\end{eqnarray*}
where $\bigcup_{i = 1}^{\hat{n}(t)} \bigcup^{o(\hat{s}_{i}(t), t)} \{
\hat{s}_{i}(t) \} = \{ s_{1}^{\prime}, \cdots,
s_{\tilde{n}(t)}^{\prime} \}$ and $m_{t}^{(i)}(s_{i}^{\prime}) =
\tilde{s}_{i}(t)$ for each $i = 1, \cdots, \tilde{n}(t)$.

\subsection{Equation governing the dynamics of $q(s, t)$}
\label{Equation_governing_the_dynamics_of_q(s,_t)}

In this subsection, we derive a partial differential equation that
describes the time evolution of the relative frequency distribution
$q(s, t)$ of $S(t)$ formulated above, letting $\Delta t \to 0$.
We suppose that there exists an upper limit to the length of a
sequence inserted between two nucleotides (two letters of $A$) in a
DNA sequence by stochastic mutation and denote it by $c$.
We set $\ell(s) = c(|s| + 1) + |s|$ for each $s \in A^{\ast}$ and
$V(s, d) = \{ s^{\prime} \in A^{\ast}: d_{L}(s, s^{\prime}) = d \}$
for $d \in \mathbb{N}$.
We put $W(s) = \{ s^{\prime} \in A^{\ast}: d_{L}(s^{\prime}, s) \leq
\ell(s^{\prime}) \}$ and $W(s, d) = \{ s^{\prime} \in W(s):
d_{L}(s^{\prime}, s) = d \}$.
$W(s)$ is the set of sequences that can produce offspring sequences
that are equal to $s$ by mutation, and we have $W(s) = \cup_{0 \leq d <
\infty} W(s, d)$.
We set $ x(s, t) = |\{ s^{\prime} \in S(t): s^{\prime} = s \}|$,
$\hat{x}(s, t) = |\{ s^{\prime} \in \hat{S}(t): s^{\prime} = s \}|$,
and $y(s, t) = \hat{x}(s, t)/x(s, t)$ for $t \in [0, \infty)$.
${}_{n} C_{r}$ represents the number of combinations of $n$ items
taken $r$ at a time.

\begin{proposition}[Evolutionary equation of a DNA population]
\label{prop:Evolutionary_equation_of_a_DNA_population}
If there exists the limit $b(s, t)$ of $y(s, t)$ letting $y(s, t),
\Delta t \to 0$ with the ratio $y(s, t)/\Delta t$ constant,
the time evolution of the relative frequency distribution $q(s, t)$ of
$S(t)$ is described by the partial differential equation
\begin{eqnarray}
\frac{\partial q(s, t)}{\partial t} &=&
o(s, t) b(s, t) q(s, t) (1 - \pi)^{\ell(s)}
\nonumber \\
& & {} + \sum_{1 \leq d < \infty} \sum_{s^{\prime} \in W(s, d)}
o(s^{\prime}, t) b(s^{\prime}, t) q(s^{\prime}, t)
\frac{{}_{\ell(s^{\prime})} C_{d} \pi^{d} (1 - \pi)^{\ell(s) -
d}}{|V(s^{\prime}, d)|}
- b(s, t) q(s, t)
\label{eq:evolutionary_equation}
\end{eqnarray}
for any $s \in A^{\ast}$ and $t \in [0, \infty)$.
\end{proposition}

\begin{proof}
In the model of the time evolution of $S(t)$ formulated above, the
relative frequency $q(s, t + \Delta t)$ of any $s \in A^{\ast}$ at
time $t + \Delta t$ consist of the following three components.
(i) the relative frequency of sequences equal to $s$ and belonging to
$S(t) \smallsetminus \hat{S}(t)$,
(ii) the relative frequency of sequences equal to $s$ and belonging to
$\hat{S}(t)$,
and (iii) the relative frequency of sequences equal to $s$ that are
produced by sequences different from $s$ and belonging to
$\hat{S}(t)$.
An offspring sequence of any $s \in S(t)$ is obtained by performing
the edit operation corresponding to Levenshtein distance one on $s$ at
most $\ell(s)$ times,
and the probabilities that the edit operation is performed and that it
is not performed are equal to $\pi$ and $1 - \pi$, respectively (see
the definitions of $c$ and $\ell(s)$ above and the definition of
mutation probability $\pi$ in the sixth paragraph of Section
\ref{sec:Formulation_of_the_model}).
In addition, there exist $|V(s, d)|$ sequences to which the
Levenshtein distance from $s$ is equal to $d$.
Therefore, the probability that an offspring sequence of $s$ is a
sequence to which the Levenshtein distance from $s$ is equal to $d$
satisfying $0 \leq d \leq \ell(s)$ is equal to ${}_{\ell(s)}C_{d}
\pi^{d} (1 - \pi)^{\ell(s) - d}/|V(s, d)|$.
Furthermore, the relative frequency in $S(t)$ of sequences that are
equal to $s$ and produce offspring and die in the unit time $[t, t +
\Delta t]$ is $y(s, t) q(s, t)$.
Thus, using the strong law of large numbers, the above (iii) is
represented as
\[
\sum_{1 \leq d < \infty} \sum_{s^{\prime} \in W(s, d)} o(s^{\prime},
t) y(s^{\prime}, t) q(s^{\prime}, t) \frac{{}_{\ell(s^{\prime})} C_{d}
\pi^{d} (1 - \pi)^{\ell(s^{\prime}) - d}}{|V(s^{\prime}, d)|}.
\]
(i) and (ii) are easy to represent and we obtain
\begin{eqnarray}
q(s, t + \Delta t) &=& \frac{x(s, t) - \hat{x}(s, t)}{x(s, t)} q(s, t)
+ o(s, t) y(s, t) q(s, t) (1 - \pi)^{\ell(s)} \nonumber \\
& & {} + \sum_{1 \leq d < \infty} \sum_{s^{\prime} \in W(s,
d)} o(s^{\prime}, t) y(s^{\prime}, t) q(s^{\prime}, t) \frac{{}_{\ell(s^{\prime})}
C_{d} \pi^{d} (1 - \pi)^{\ell(s^{\prime}) - d}}{|V(s^{\prime}, d)|}.
\label{eq:balance_equation}
\end{eqnarray}
Expanding the left hand side of Equation (\ref{eq:balance_equation})
as $q(s, t + \Delta t) = q(s, t) + \Delta t \partial q(s, t)/\partial
t$ in a Taylor series with respect to $t$ and rearranging the
equation provides
\begin{eqnarray*}
\frac{\partial q(s, t)}{\partial t} &=&
\frac{1}{\Delta t} o(s, t) y(s, t) q(s, t) (1 - \pi)^{\ell(s)}\\
& & {} + \frac{1}{\Delta t} \sum_{1 \leq d < \infty} \sum_{s^{\prime} \in W(s, d)}
o(s^{\prime}, t) y(s^{\prime}, t) q(s^{\prime}, t) \frac{{}_{\ell(s^{\prime})}
C_{d} \pi^{d} (1 - \pi)^{\ell(s) - d}}{|V(s^{\prime}, d)|}
- \frac{1}{\Delta t} y(s, t) q(s, t).
\end{eqnarray*}
Letting $y(s, t), \Delta t \to 0$ and $y(s^{\prime}, t), \Delta t \to
0$ with the ratios $y(s, t)/\Delta t$ and $y(s^{\prime}, t)/\Delta t$
constant, we obtain Equation (\ref{eq:evolutionary_equation}) by the
assumption of the proposition.
\end{proof}

$q(s^{\prime}, t)$ in the third term of the right hand side of
Equation (\ref{eq:balance_equation}) cannot be expanded in a Taylor
series with respect to $s^{\prime}$ because $s$ is a discrete
variable, and consequently, the resultant equation, Equation
(\ref{eq:evolutionary_equation}), does not include partial
derivatives with respect to $s^{\prime}$ in the right hand side,
unlike the diffusion equation.



\subsection{Models of the selection pressure and mutation}
\label{subsec:Formulation_of_the_model_continuation}

In this subsection, we formulate models of the selection pressure
$p(s, E)$ and stochastic mutation $m$ to complete the model
construction, which will be necessary in performing numerical
experiments based on the model constructed in Subsection
\ref{subsec:Evolutionary_model_of_a_DNA_population}.

We suppose that the critical value of selection pressure is $p_{E} =
1$ for a given environment $E \in \mathcal{E}$ without the loss of
generality.
Let $D_{j} \subset A^{\ast}$ for each $j = 1, \cdots, k$ and $D_{j}
\cap D_{j^{\prime}} = \emptyset$ if $j \neq j^{\prime}$.
For $\lambda_{j} \in A^{\ast}$ and $\rho_{j} \in (0, \infty)$, we
define the function $\phi( \ \cdot \ , \lambda_{j}, \rho_{j}) : A^{\ast}
\to [0, 1]$ as
\[
\phi(s, \lambda_{j}, \rho_{j}) = \frac{1}{(\rho_{j} + 1) \bigl|
V(\lambda_{j}, d_{L}(s, \lambda_{j})) \bigr|}
\left( \frac{\rho_{j}}{\rho_{j} + 1} \right)^{d_{L}(s, \lambda_{j})}.
\]
We denote the truncated function of $\phi(s, \lambda_{j}, \rho_{j})$
on $D_{j}$ by $\phi_{D_{j}}(s, \lambda_{j}, \rho_{j})$, i.e.,
$\phi_{D_{j}}(s, \lambda_{j}, \rho_{j}) = \phi(s, \lambda_{j},
\rho_{j})\\/\sum_{s \in D_{j}} \phi(s, \lambda_{j}, \rho_{j})$ if $s
\in D_{j}$ and $\phi_{D_{j}}(s, \lambda_{j}, \rho_{j}) = 0$ if $s
\notin D_{j}$.
We set $\lambda = (\lambda_{1}, \cdots, \lambda_{k}), \rho =
(\rho_{1}, \cdots, \rho_{k})$, $D = (D_{1}, \cdots, D_{k})$, and
$w = (w_{1}, \cdots, w_{k})$ for $w_{1}, \cdots, w_{k} \in (0, 1)$.
$\phi_{D}(s, \lambda, \rho, w)$ represents the mixture model of
$\phi_{D_{1}}(s, \lambda_{1}, \rho_{1}), \cdots, \phi_{D_{k}}(s,
\lambda_{k}, \rho_{k})$ with mixture coefficients $w_{1}, \cdots,
w_{k}$.
We define the selection pressure $p(s, E)$ exerted by the environment
$E$ on $s \in A^{\ast}$ as
\begin{eqnarray*}
p(s, E) = \left\{
\begin{array}{ll}
1 - \phi_{D}(s, \lambda, \rho, w) & \mbox{if } s \in
\displaystyle{\bigcup_{j = 1}^{k}} D_{j},\\
\mbox{arbitrary number } \geq 1 & \mbox{otherwise}.
\end{array}
\right.
\end{eqnarray*}
In this setting, $D_{1}, \cdots, D_{k}$ are reproductive sequence
domains under $E$ and the truncated function of $\phi_{D}(s, \lambda,
\rho, w)$ on $\hat{S}(t)$ (or $S(t)$) provides the relative fitness of
each sequence in $\hat{S}(t)$ (or $S(t)$).

We next formulate the mutation operator $m$.
$m$ cannot be written by an analytic expression.
We construct an algorithm generating a new mutated sequence based on a
given sequence that satisfies the conditions (i) to (iii) described in
the fifth paragraph of Section~\ref{sec:Formulation_of_the_model}.
The algebraic structure of $A^{\ast}$ as a monoid is required for
constructing $m$.
We introduce the empty letter $e$ and set $e \in A^{\ast}$.
We have $|e| = 0$.
We define the concatenation $s \cdot s^{\prime}$ as $s \cdot
s^{\prime} = x_{1} \cdots x_{|s|} x_{1}^{\prime} \cdots
x_{|s^{\prime}|}^{\prime}$ for $s = x_{1} \cdots x_{|s|}, s^{\prime} =
x_{1}^{\prime} \cdots x_{|s^{\prime}|}^{\prime} \in A^{\ast}$.
The concatenation $\cdot$ is an interior operation in $A^{\ast}$, and
$A^{\ast}$ forms a monoid with an identity $e$.
The metric topology of $d_{L}$ on $A^{\ast}$ is the discrete topology.
Therefore, $A^{\ast}$ forms a topological monoid.
DNA replication is performed by DNA polymerase's adding nucleotides to
one end of the newly forming DNA strand based on the order of
nucleotides in a template DNA sequence, and mutations are errors that
occurred in this process.
Therefore, we define the mutation operator $m$ by combining the
following algorithm and models of insertion, deletion, and
substitution of nucleotides provided by Definitions
\ref{def:Model_1_Insertion)} and
\ref{def:Model_2_Substitution_and_deletion}.
We set $\bar{A} = A \cup \{ e \}$.

\begin{algorithm}[htbp]
\caption{Generate a mutated sequence}
\label{alg:Generate_a_mutated_sequence}
\begin{algorithmic}[1]
\Require $s \in A^{\ast}, c \in \mathbb{Z}^{+}, \pi \in (0, 1)$

\State $\ell \gets |s|$; $s^{\prime} \gets e$

\For{$i = 1, \cdots, \ell$}

    \For{$j = 1, \cdots, c$}

        \State Choose $y_{1} \in \bar{A}$ according to Model $1$ \Comment{Whether a substring is inserted or not before each letter is determined.}

        \If{$y_{1} \neq e$}
             \State $s^{\prime} \gets s^{\prime} \cdot y_{1}$ \Comment{Insertion}
%
        \EndIf
    \EndFor

    \State Choose $y_{2} \in \bar{A}$ according to Model $2$ \Comment{Whether each letter is substituted or deleted, or no edit is performed is determined.}

    \If{$y_{2} \neq e$}
        \State $s^{\prime} \gets s^{\prime} \cdot y_{2}$ \Comment{Substitution}
    \Else
        \State $s^{\prime} \gets s^{\prime}$ \Comment{Deletion}

%
    \EndIf

\EndFor

\For{$k = 1, \cdots, c$}
\State Choose $y_{3} \in \bar{A}$ according to Model $1$
\If{$y_{3} \neq e$}
    \State $s^{\prime} \gets s^{\prime} \cdot y_{3}$ \Comment{Insertion}
\EndIf
\EndFor
\State \textbf{return} $s^{\prime}$
\end{algorithmic}
\end{algorithm}


\begin{definition}{Model 1 (Insertion)}
\label{def:Model_1_Insertion)}
Suppose that we are at the $i_{0}$th step of the $i$-loop of
Algorithm 1.
If (i) $i_{0} \geq 2$, (ii) the $(i_{0} - 1)$th letter of the input
string $s$ is equal to $x \in A$, (iii) the letter $x$ was deleted
(i.e., $y_{2} = e$) at the $(i_{0} - 1)$th step of the $i$-loop, and
(iv) no letters except $x$ was inserted (i.e., $y_{1} \notin A
\smallsetminus \{ x \}$) at all previous steps of the $j$-loop,
generate $y \in \bar{A}$ according to the probability function
\[
\psi_{I}(y; \pi) = \left\{
\begin{array}{ll}
1 - \pi & \mbox{if } y = e,\\
\pi/3 & \mbox{if } y = A \smallsetminus \{ x \}, \\
0 & \mbox{if } y = x
\end{array}
\right.
\]
on $\bar{A}$, otherwise generate $y \in \bar{A}$ according to
\[
\psi_{I}(y; \pi) = \left\{
\begin{array}{ll}
1 - \pi & \mbox{if } y = e, \label{eq:character_generation_1_2} \\
\pi/4 & \mbox{if } y \in A.
\end{array}
\right.
\]
\end{definition}
In the $k$-loop after the $i$-loop, the above (ii) to (iv) are changed
into the following (ii$^{\prime}$) to (iv$^{\prime}$), respectively.
(ii$^{\prime}$) The last letter of the input string $s$ is
equal to $x \in A$.
(iii$^{\prime}$) The letter $x$ was deleted (i.e., $y_{2} = e$) at the
last step of the $i$-loop. And (iv$^{\prime}$) no letters except
$x$ was inserted (i.e., $y_{3} \notin A \smallsetminus \{ x \}$) at
all previous steps of the $k$-loop.


\begin{definition}{Model 2 (Substitution and deletion)}
\label{def:Model_2_Substitution_and_deletion}
Suppose that we are at the $i_{0}$th step of the $i$-loop of
Algorithm~1.
If (i) the $i_{0}$th letter of the input string $s$ is equal to $x \in
A$ and (ii) the last letter inserted during the $j$-loop is equal to
$x$, generate $y \in \bar{A}$ according to the probability function
\[
\psi_{DS}(y; \pi) = \left\{
\begin{array}{ll}
1 - \pi & \mbox{if } y = x,\\
\pi/3 & \mbox{if } y \in A \smallsetminus \{ x \}, \\
0 & \mbox{if } y = e
\end{array}
\right.
\]
on $\bar{A}$, otherwise generate $y \in \bar{A}$ according to
\[
\psi_{DS}(y; \pi) = \left\{
\begin{array}{ll}
1 - \pi & \mbox{if } y = x,\\
\pi/4 & \mbox{if } y \in \bar{A} \smallsetminus \{ x \}.
\end{array}
\right.
\]
\end{definition}


In actual replication of DNA, a complementary sequence including
mutations is composed from a template sequence (in this case, a parent
sequence).
Therefore, to be exact, an input string $s$ of Algorithm 1 should be
interpreted as a complementary sequence including no mutations of a
parent sequence.
However, no problems occur in numerical experiments even if $s$ is
considered a parent sequence.
It is verified that Algorithm 1 as a mapping from $A^{\ast}$ to
$A^{\ast}$ satisfies conditions (i) to (iii) described in the fifth
paragraph of Section~\ref{sec:Formulation_of_the_model}.
Therefore, we define the mutation operator $m$ as the operator that
maps an element of $A^{\ast}$ to other element according to
Algorithm 1.




\section{Specification of the problem}
\label{sec:Specification_of_the_problem}

In this section, we introduce several definitions and specify problems
to be considered in the following sections.
We begin by defining the differentiation of a population of DNA
sequences.
Any subset of $A^{\ast}$ is open because the topology of $A^{\ast}$
with respect to $d_{L}$ is the discrete topology.
Thus, $D \subset A^{\ast}$ is a domain in the mathematical sense if it
cannot be represented as the union of two or more nonempty disjoint
subsets of $A^{\ast}$.

\begin{definition}[Reproductive sequence domain]
\label{def:reproductive_sequence_domain}
We say that a domain $D \subset A^{\ast}$ is a reproductive sequence
domain under an environment $E \in \mathcal{E}$ if $p_{E} - p(s; E) >
0$ holds for any $s \in D$.
\end{definition}
A reproductive sequence domain is not a geographical region but a
subset of $A^{\ast}$ of which DNA sequences that can produce offspring
sequences in a given environment are composed.

For asexual populations such as microbial populations, we say that
speciation occurred if DNAs of a new population are sufficiently far
from those of an original population (homologies between any pair
of DNAs or specific genes in the two populations are lower than a
certain threshold).
On the other hand, for sexual populations, we say that speciation
occurred if any pair of a male and female in a new and original
populations cannot produce offspring that have reproductive
capacities.
It is frequently difficult to test whether a biological population
differentiated in natural environments according to this definition
for the latter populations.
However, the reason why any pair of a male and female in the two
populations cannot produce offspring that have reproductive capacities
would be that DNA sets of the two populations are sufficiently
different.
Therefore, in this study, we formulate the differentiation of a
population as follows:

\begin{definition}[DNA population differentiation]
We suppose that there exist two reproductive sequence domains $D_{1}, D_{2}
\subset A^{\ast}$ under the environment $E \in \mathcal{E}$.
Let $S(t)$ be a population of DNA sequences at time $t$
from an ancestor sequence born in $D_{1}$.
We say that $S(t)$ differentiated at time $t_{0} \in [0, \infty)$ if
(i) an offspring sequence $s^{\prime} \in D_{2}$ was generated from
a sequence $s \in S(t_{0})$ at time $t_{0}$ by mutation and (ii)
$S(t) \cap D_{2} = \emptyset$ holds for any $t < t_{0}$.
\end{definition}

Lastly, we define the equilibrium state of a DNA population $S(t)$.
The extent to which two biological communities differ is called
$\beta$ dissimilarity or $\beta$ diversity and have studied in
ecology.
Here, the difference includes that of a population between two time
points and that between two different populations.
The $\beta$ dissimilarity $d_{\beta}(S_{1}, S_{2})$ between two
populations $S_{1}$ and $S_{2}$ of DNA sequences that two biological
communities have was introduced and applied in \cite{Koyano_2014}
after investigating a method for estimating it.
Using $d_{\beta}$, the evolutionary rate of a DNA population $S(t)$ at
time $t$ can be defined as
\[
\dot{S}(t) = \lim_{\Delta t \to 0} \frac{d_{\beta}(S(t + \Delta t),
S(t))}{\Delta t}.
\]
(Lots of studies have been conducted with respect to evolutionary rate
in evolutionary biology. See, for example, \cite{Thorne_2002}.)
Therefore, according to the traditional manner, the equilibrium state
of a population $S(t)$ can be defined as follows: $S(t)$ is in the
equilibrium state during $[t_{0}, t_{1}] \subset [0, \infty)$ if
$\dot{S}(t) = 0$ holds for any $t \in [t_{0}, t_{1}]$.
However, this definition of equilibrium state is not useful for
$S(t)$ because there always exists a possibility that offspring
sequences of sequences in $S(t)$ include mutations that occur
randomly.
Therefore, in this study, we consider a sort of equilibrium state
defined as follows:

\begin{definition}[near equilibrium state]
\label{def:near_equilibrium_state}
A population $S(t)$ is in the near equilibrium state during $[t_{0},
t_{1}] \subset [0, \infty)$ if $S(t)$ satisfies the following
conditions.
\begin{enumerate}
\item There exists $M \subset D$ satisfying $M \neq \emptyset$ such
  that if $t \in [t_{0}, t_{1}]$, then $S(t) \supset M$.

\item $o(s, t) = 0$ holds for any $t \in [t_{0}, t_{1}]$ and $s \in
  S(t) \cap M^{c}$.
\end{enumerate}
\end{definition}
In a state of near equilibrium, the population $S(t)$ always includes
sequences equal to $s$ for any $s \in M$, the relative frequency $q(s,
t)$ of them is updated at each time $t$ according to Equation
(\ref{eq:evolutionary_equation}), and, even if sequences were born
outside $M$ by mutation, they cannot produce offspring.
This is a minimum variation in $S(t)$ under stochastic mutation
and selection pressure from the environment.

In Section \ref{sec:Theoretical_analysis}, we first consider the
problem of under what conditions $S(t)$ differentiates or not in the
sympatric setting rather than in the allopatric setting, in other
words, without supposing that $S(t)$ is divided into two
subpopulations that live in different environments $E$ and
$E^{\prime}$ and therefore are under selection pressures based on
different selection pressure functions $p(s, E)$ and $p(s,
E^{\prime})$.
We subsequently consider whether $S(t)$ can reach and maintain a state
of near equilibrium under some condition.
Furthermore, we calculate the differentiation probability of $S(t)$
and the time expected to elapse before $S(t)$ differentiates.

\section{Theoretical analysis}
\label{sec:Theoretical_analysis}

In this section, we address the problems specified in the previous
section.
We first consider in what situations a population of DNA sequences
differentiates or not as time elapses.
We set $U(s, d) = \{ s^{\prime} \in A^{\ast}: d_{L} (s, s^{\prime})
\leq d \}$ for $s \in A^{\ast}$ and $d \in \mathbb{N}$.
We denote the almost sure convergence by $\asconv$.

\begin{proposition}[Condition for differentiation]
\label{prop:Condition_for_differentiation}
We suppose that there exist two reproductive sequence domains $D_{1},
D_{2} \subset A^{\ast}$ under an environment $E \in \mathcal{E}$.
Let $S(t)$ be a population of DNA sequences at time $t$ from an
ancestor sequence born in $D_{1}$.
If the following conditions (1) to (3) are satisfied, the probability
that $S(t)$ differentiates into $D_{2}$ approaches one as $t \to
\infty$.
Conversely, if the negation of the condition (1) holds, then $S(t)$
does not differentiate for any $t \in [0, \infty)$.
\begin{enumerate}
\item There exist $s_{1} \in D_{1}$ and $s_{2} \in D_{2}$ such that
  $d_{L}(s_{1}, s_{2}) \leq \ell(s_{1})$ holds.

\item $p(s, E) = 1/|D_{1}|$ for any $s \in D_{1}$.

\item $n(t) \geq 1$ for any $t \in [0, \infty)$.
\end{enumerate}
\end{proposition}

The former statement in the above proposition indicates that, wherever
in a reproductive sequence domain $D_{1}$ an ancestor sequence was
born, the population differentiates with high probability if there
exists another reproductive sequence domain near to $D_{1}$ (condition
(1)) and the selection pressure is close to the uniform distribution
(condition (2)).
The conditions of Proposition \ref{prop:Condition_for_differentiation}
do not include those of the split of one population into two
geographically isolated populations for reproductive isolation.
They are conditions for reproductive isolation to occur in a
population living in one habitat.
Therefore, if the differentiation described in the proposition is
speciation, it is not allopatric but sympatric speciation (see, for
example,
\cite{Johnson_1996,Kondrashov_1999,Dieckmann_1999,Kawata_2002} for
articles on theoretical studies on sympatric speciation and
\cite{Bolnick_2007} for a review).
Even if there exists another reproductive sequence domain nearby,
$S(t)$ does not necessarily differentiate, as demonstrated in
Corollary \ref{cor:near_equilibrium_state-1} below.

\begin{proof}
The set of sequences that can be produced by $s \in D_{1}$ is $U(s,
\ell(s))$, and thus, the latter part of the proposition is trivial.
Therefore, we demonstrate the former part.

\textbf{(Step 1)}
We denote the numbers of sequences in the population $\{
S(t^{\prime}): 0 \leq t^{\prime} \leq t \}$ that die by time $t$ and
that are equal to $s$ and die by $t$ by $\kappa(t)$ and $\kappa(s,
t)$, respectively, for $t \in [0, \infty)$ and $s \in D_{1}$.
Noting that the number of strings that are equal to or less than
$\ell$ in length is equal to $\sum_{\ell^{\prime} = 1}^{\ell}
4^{\ell^{\prime}} + 1$ for $\ell \in \mathbb{N}$, we have $|A^{\ast}|
< \infty$ from the definition of a string (see the second paragraph of
Section \ref{section:Introduction}).
Therefore,
\begin{equation}
|D_{1}| < \infty
\label{eq:number_of_elements_of_D_1}
\end{equation}
holds from $D_{1} \subset A^{\ast}$.
Using the conditions (3) of the proposition and C$_{1}$ described in
the third paragraph of Section \ref{sec:Formulation_of_the_model}, we
obtain
\begin{equation}
\kappa(t) \to \infty
\label{eq:sum_of_the_sizes_of_populations}
\end{equation}
as $t \to \infty$.
From Equations (\ref{eq:number_of_elements_of_D_1}) and
(\ref{eq:sum_of_the_sizes_of_populations}), there exists $s \in D_{1}$
such that $\kappa(s, t) \to \infty$ as $t \to \infty$.
Choosing such a $s \in D_{1}$, we have $o(s, t) \to \infty$ as $t \to
\infty$ from the conditions (2) and (3) and the definition of $o(s,
t)$ in the fifth paragraph of Section
\ref{sec:Formulation_of_the_model}.

\textbf{(Step 2)}
We arbitrarily choose $s^{\prime} \in U(s, \ell(s))$ for $s \in D_{1}$
chosen in Step 1.
We regard the production of an offspring sequence by a sequence in
$S(t)$ that is equal to $s$ as a trial that is a success if the
sequence produces an offspring sequence equal to $s^{\prime}$ by
mutation and is a failure otherwise.
Let $p_{i}$ represent the success probability of the $i$th trial,
counting from the first trial performed by a sequence in $\{ S(t): t
\geq 0 \}$ that are equal to $s$.
By the definition of the mutation operator $m$ in the sixth paragraph
of Section \ref{sec:Formulation_of_the_model}, $p_{i}$ is constant,
not depending on $i$ (hence, $p_{i}$ is written as $p$ hereafter), and
the trials are independent.
Therefore, this trial is a Bernoulli trial.
We have
\begin{equation}
p = {}_{\ell(s)}C_{d_{L}(s, s^{\prime})} \pi^{d_{L}(s,
s^{\prime})} (1 - \pi)^{\ell(s) - d_{L}(s, s^{\prime})} > 0
\label{eq:p_is_positive}
\end{equation}
from $d_{L}(s, s^{\prime}) \leq \ell(s)$ and $\pi \in (0, 1)$.
Let $X_{i}$ be a Bernoulli variable that takes one if the $i$th trial
is a success and zero otherwise.
Using the result obtained in Step 1, the strong law of large numbers,
and Equation (\ref{eq:p_is_positive}), we obtain $\sum_{i = 1}^{n}
X_{i} / n \asconv p$ and consequently $\sum_{i = 1}^{n} X_{i} \asconv
\infty$ as $n \to \infty$.
This means that sequences in $\{ S(t): t \geq 0 \}$ that are
equal to $s$ produce offspring sequences equal to $s^{\prime}$
infinite times with probability one as $t \to \infty$.
The above discussion is independent of the choice of $s^{\prime} \in
U(s, \ell(s))$, and therefore, we have $\kappa(s^{\prime}, t) \asconv
\infty$ as $t \to \infty$ for any $s^{\prime} \in U(s, \ell(s))$.
Thus, we have $o(s^{\prime}, t) \asconv \infty$ as $t \to \infty$ from
the conditions (2) and (3) and the definition of $o(s, t)$.
Repeating the same discussion as above provides $o(s^{\prime \prime},
t) \asconv \infty$ as $t \to \infty$ for any $s^{\prime \prime} \in
D_{1}$ because $D_{1}$ is a domain in the mathematical sense.

\textbf{(Step 3)}
From the result of 2, we have $o(s_{1}, t) \asconv \infty$ as $t \to
\infty$ for $s_{1} \in D_{1}$ satisfying the condition (1).
Therefore, conducting the same discussion as in Step 2, from the
condition (1), we observe that sequences in $\hat{S}(t)$ that are
equal to $s_{1}$ produce sequences in $D_{2}$ infinite times with
probability one as $t \to \infty$.
This completes the proof of the former part of the proposition.
\end{proof}

$\{ s \in S(t): p(s, E) \geq p_{E} \}$ is the set of sequences in a
DNA population $S(t)$ that were produced outside all reproductive
sequence domains by mutation and, therefore, cannot produce offspring.
If $\{ s \in S(t): p(s, E) < p_{E} \} \subsetneq M$ (proper subset)
holds for $M \subset A^{\ast}$, we write $S(t) \prec M$.
$\mathbb{Z}^{+}$ represents the set of positive integers.
As the following proposition describes, under certain conditions, it
is determined whether sympatric speciation occurs or a population
continues to move around randomly in a subset of $A^{\ast}$, depending
on the existence or nonexistence of other reproductive sequence domain
nearby.

\begin{corollary}
\label{cor:nonequilibrium_state}
We suppose that there exist two reproductive sequence domains $D_{1}, D_{2}
\subset A^{\ast}$ under an environment $E \in \mathcal{E}$.
Let $S(t)$ be a population of DNA sequences at time $t$ from an
ancestor sequence born in $D_{1}$.
If the following conditions (1) to (3) are satisfied, the probability
that $S(t)$ differentiates into $D_{2}$ approaches one as $t \to
\infty$.
On the other hand, if the conditions (2) and (3) and the negation of
the condition (1) are satisfied, then $S(t)$ continues to move around
randomly in $D_{1}$ as $t \to \infty$.
\begin{enumerate}
\item There exist $s_{1} \in D_{1}$ and $s_{2} \in D_{2}$ such that
  $d_{L}(s_{1}, s_{2}) \leq \ell(s_{1})$ holds.

\item $p(s, E) = 1/|D_{1}|$ for any $s \in D_{1}$.

\item There exist $t_{0} \in [0, \infty)$ and $n^{\ast} \in
  \mathbb{Z}^{+}$ such that $1 \leq n(t) \leq n^{\ast}$ for any $t
  \geq t_{0}$ and $n^{\ast} < |D_{1}|$ hold.
\end{enumerate}
\end{corollary}

\begin{proof}
The former part of the proposition is immediately obtained from the
former part of Proposition \ref{prop:Condition_for_differentiation}.
Therefore, we demonstrate the latter part.
Using the conditions (2) and (3), we have $\kappa(s, t) \asconv
\infty$ as $t \to \infty$ for any $s \in D_{1}$ by the same discussion
in Steps 1 and 2 in the proof of Proposition
\ref{prop:Condition_for_differentiation}.
However, from the negation of the condition (1) and Proposition
\ref{prop:Condition_for_differentiation}, $S(t)$ does not
differentiate into $D_{2}$ for any $t \in [0, \infty)$.
Consequently, $S(t) \prec D_{1}$ holds from the condition (3).
Therefore, we arbitrarily choose $s_{0} \in D_{1} \smallsetminus S(t)$.
Conducting the same discussion in Steps 1 and 2 in the proof of
Proposition \ref{prop:Condition_for_differentiation} from the
conditions (2) and (3) again, we obtain $\kappa(s_{0}, t) \asconv
\infty$ as $t \to \infty$.
Choosing a sequence from $D_{1} \smallsetminus S(t)$ arbitrarily for
each $t \in [0, \infty)$ and repeating the above discussion,
we see that $S(t)$ continues to move around randomly in $D_{1}$ as $t
\to \infty$.
\end{proof}

From the above proposition, we found that under the listed conditions
a population of DNA sequences maintains a state of nonequilibrium.
There exist species called living fossils whose phenotypes
have been nearly unchanged for a long time.
Therefore, we next consider whether it is possible that a DNA
population have undergone little change for a long period.
This problem is restated as whether a DNA population can maintain a
state of near equilibrium under some condition.
There exist examples (e.g., shark) in which phenotypes have been
nearly unchanged, but DNA sequences have changed more than expected.
Thus, depending on whether the above problem is positively solved or
negatively solved, it is determined whether there exist two possible
explanations or there exists only one explanation at the sequence
level for the phenomena that phenotypes have been nearly unchanged for
a long time.

To solve the above problem we introduce the following definition:
$\tilde{n}(t)$ is consistent with $p(s, E)$ if $p(s, E) =
p(s^{\prime}, E)$ implies $o(s, t) = o(s^{\prime}, t)$ for any $s,
s^{\prime} \in \hat{R}(t)$.
In addition, we introduce the following condition C${}_{3}$:
For any $t \in [0, \infty)$ and $s \in S(t)$,
\[
o(s, t) b(s, t) q(s, t) (1 - \pi)^{\ell(s)} > o(s, t) b(s, t) q(s, t)
\sum_{d = 1}^{\ell(s)} {}_{\ell(s)}C_{d} \pi^{d} (1 - \pi)^{\ell(s) -
d}
\]
holds.
C${}_{3}$ means that the mutation probability $\pi$ is sufficiently
small, such that the number of offspring sequences without mutations
that each sequence produces is less than that of offspring sequences
with mutations.
From here to the end of Corollary \ref{cor:near_equilibrium_state-1}
below, we consider the case of $\gamma = 1$.
Therefore, we have $R(t) = \hat{R}(t)$.
We set $\langle R(t) \rangle = \{ s \in R(t): o(s, t) \neq 0 \}$ for
any $t \in [0, \infty)$.
As the following proposition shows, if the size of a population that
an environment can accommodate decreases because of environmental
change and so on, and subsequently the population size remains
constant, a DNA population can maintain a state of near equilibrium in
the case where the mutation probability is sufficiently low.

\begin{proposition}[Conditions for near equilibrium state]
\label{prop:Yano's_problem}
Let $D \subset A^{\ast}$ be a reproductive sequence domain under an
environment $E \in \mathcal{E}$ and $S(t)$ be a population of DNA
sequences at time $t$.
We suppose that $t_{0}, t_{1}, t_{2} \in [0, \infty)$ are time points
satisfying $t_{0} < t_{1} < t_{2}$.
If the following conditions are satisfied, $S(t)$ maintains a state of
near equilibrium during $[t_{1}, t_{2}] \subset [0, \infty)$.
\begin{enumerate}
\item $n(t)$ monotonically decreases with $t \in [t_{0}, t_{1}]$.

\item There exists $n^{\ast} \in \mathbb{Z}^{+}$ such that $n(t) =
  n^{\ast}$ holds for any $t \in [t_{1}, t_{2}]$.

\item $n^{\ast}$ is consistent with $p(s, E)$ (this condition makes
  sense from the condition (2) above and the assumption of $\gamma =
  1$).

\item Condition C${}_{3}$ is satisfied.

\item For any $s \in \langle R(t_{1}) \rangle$, there does not exist
  $s^{\prime} \in (\bigcup_{t_{0} \leq t \leq t_{1}} R(t))^{c}$ that
  satisfies
\[
d_{L}(s, s^{\prime}) \leq \ell(s), \quad p(s^{\prime}, E) \leq \max \{
p(s^{\prime \prime}, E): s^{\prime \prime} \in \langle R(t_{1})
\rangle \}.
\]
\end{enumerate}
\end{proposition}

\begin{proof}
First, we have
\begin{equation}
\langle R(t_{1}) \rangle \neq \emptyset
\label{eq:angle_R_t_1_is_not_empty}
\end{equation}
from the condition (2).
Using the conditions (2) and (3) provides $o(s, t) \neq 0$ for any $s
\in \langle R(t_{1}) \rangle$ and $t \in [t_{1}, t_{2}]$.
Thus, $s \in R(t)$ holds from the condition (4), and therefore,
we have
\begin{equation}
R(t) \supset \langle R(t_{1}) \rangle.
\label{eq:R_t_includes_angle_R_t_1}
\end{equation}
From the condition (1), $\bigcup_{t_{0} \leq t \leq t_{1}} \langle
R(t) \rangle$ is a set obtained by adding sequences that can produce
offspring when the population size is greater than $n^{\ast}$ to those
that can produce offspring when it is equal to $n^{\ast}$.
Hence, using the condition (2), we obtain
\begin{equation}
o(s^{\prime \prime}, t) = 0
\label{eq:o_s_prime_prime_t_equal_0}
\end{equation}
for any $s^{\prime \prime} \in \bigcup_{t_{0} \leq t \leq t_{1}}
\langle R(t) \rangle \smallsetminus \langle R(t_{1}) \rangle$ and $t
\in [t_{1}, t_{2}]$.
Noting the condition (5), we observe that even if a sequence in $S(t)$
produces a sequence $s^{\prime} \in (\bigcup_{t_{0} \leq t \leq t_{1}}
\langle R(t) \rangle)^{c}$ by mutation,
\begin{equation}
o(s^{\prime}, t) = 0
\label{eq:o_s_prime_t_equal_0}
\end{equation}
holds.
From Equations (\ref{eq:o_s_prime_prime_t_equal_0}) and
(\ref{eq:o_s_prime_t_equal_0}) and $\langle R(t_{1}) \rangle \subset
\bigcup_{t_{0} \leq t \leq t_{1}} \langle R(t) \rangle$, we have
\begin{equation}
o(s, t) = 0
\label{eq:o_s_t_equal_0}
\end{equation}
for any $s \in \langle R(t_{1}) \rangle$ and $t \in [t_{1}, t_{2}]$.
Combining Equations (\ref{eq:angle_R_t_1_is_not_empty}),
(\ref{eq:R_t_includes_angle_R_t_1}), and (\ref{eq:o_s_t_equal_0})
completes the proof of the proposition.
\end{proof}

Using Proposition \ref{prop:Yano's_problem}, we see that it is
possible that a DNA population maintains a state of near equilibrium
and does not differentiate even if there exists other reproductive
sequence domain nearby.

\begin{corollary}
\label{cor:near_equilibrium_state-1}
Let $D_{1}, D_{2} \subset A^{\ast}$ be reproductive sequence domains and we
consider the situation in which there exists $s \in D_{1}$ that can
produce $s^{\prime} \in D_{2}$ by mutation.
$S(t)$ represents a population of DNA sequences at time $t$.
We suppose that the conditions (1) to (5) of Proposition
\ref{prop:Yano's_problem} with $t_{2} = \infty$ are satisfied.
If $S(t)$ does not differentiate into $D_{2}$ within $[0, t_{1}]$,
then $S(t)$ no longer differentiate into $D_{2}$.
\end{corollary}

\begin{proof}
Trivial from Proposition \ref{prop:Yano's_problem}.
\end{proof}



We next calculate the probability that a population $S(t)$ of DNA
sequences at time $t$ from an ancestor sequence born in a reproductive
sequence domain $D_{1} \subset A^{\ast}$ differentiates into another
reproductive sequence domain $D_{2} \subset A^{\ast}$.
Let $\hat{S}(t, D_{2})$ be a set of $s \in \hat{S}(t)$ for which there
exists $s^{\prime} \in D_{2}$ such that $d_{L}(s, s^{\prime}) \leq
\ell(s)$ holds.
We denote a set of $s^{\prime} \in D_{2}$ satisfying $d_{L}(s,
s^{\prime}) \leq \ell(s)$ by $D_{2}(s \in \hat{S}(t))$ for each
$s \in \hat{S}(t, D_{2})$.
We put $h(s) = \min \{ d_{L}(s, s^{\prime}): s^{\prime} \in D_{2} \}$
for each $s \in \hat{S}(t, D_{2})$.
We set $z(s, d) = | \{ s^{\prime} \in D_{2}(s \in \hat{S}(t)):
d_{L}(s, s^{\prime}) = d \} |$ for any $s \in \hat{S}(t, D_{2})$ and
$d \in \{ h(s), \cdots, \ell(s) \}$.

\begin{lemma}[Probability of differentiation]
\label{lemma:prob_speciation}
We suppose that there exist two reproductive sequence domains $D_{1},
D_{2} \subset A^{\ast}$ under an environment $E \in \mathcal{E}$.
Let $S(t)$ be a population of DNA sequences at time $t$ from an
ancestor sequence born in $D_{1}$.
If there exists the limit $b(s, t)$ of $y(s, t)$ letting $y(s, t),
\Delta t \to 0$ with the ratio $y(s, t)/\Delta t$ constant for any $s
\in A^{\ast}$ and $t \in [0, \infty)$,
the probability that $S(t)$ differentiates into $D_{2}$ at time $t$ is
provided by
\begin{equation}
\zeta(t) = \sum_{s \in \hat{S}(t, D_{2})} o(s, t) b(s, t) q(s, t)
\sum_{d = h(s)}^{\ell(s)} \frac{ z(s, d)}{|V(s, d)|} {}_{\ell(s)}
C_{d} \ \pi^{d} (1 - \pi)^{\ell(s) -
d}. \label{eq:speciation_probability}
\end{equation}
\end{lemma}

\begin{proof}
Noting that only sequences in $\hat{S}(t)$ that belong to $\hat{S}(t,
D_{2})$ can produce sequences in $D_{2}$ by mutation,
that we have $h(s) \leq d_{L}(s, s^{\prime}) \leq \ell(s)$ for any $s
\in \hat{S}(t, D_{2})$ and $s^{\prime} \in D_{2}(s \in \hat{S}(t))$,
and that the ratio of sequences that belong to $D_{2}$ of sequences
$s^{\prime} \in A^{\ast}$ satisfying $d_{L}(s, s^{\prime}) = d$ for $s
\in \hat{S}(t, D_{2})$ is $z(s, d) / |V(s, d)|$,
we observe that the probability that $S(t)$ differentiates into
$D_{2}$ at time $t$ is provided by Equation
(\ref{eq:speciation_probability}).
\end{proof}

Using Proposition~\ref{lemma:prob_speciation}, we derive a formula of
the expected time for estimating how much time passes until population
differentiation.
We first consider the case where $t$ is discrete and represents a
generation number.
The number of generations that elapse before population
differentiation is the number of generations that elapse before the
first success when the trial is repeated in which the probability that
each sequence of each generation produces offspring in a reproductive
sequence domain different from that where an ancestor sequence was
born is a success probability.
This trial is a Poisson trial because the success probability is
$\zeta(t)$ provided by Equation (\ref{eq:speciation_probability}) and
varies with each trial.
Therefore, the distribution of probability that the first success is
obtained by a sequence in the $t$th generation is the distribution of
waiting time when the Poisson trial with a sequence $\{ \zeta(t): t
\in \mathbb{N} \}$ of success probabilities is repeated, and thus, its
probability function is provided by $f(t, \{ \zeta(t^{\prime}):
t^{\prime} \in \mathbb{N} \}) = \zeta(t) \prod_{t^{\prime} = 0}^{t -
1} (1 - \zeta(t^{\prime}))$.
However, it is impossible to calculate the expected value of this
probability function, i.e., the sum of the series $\E(\tau) = \sum_{t
= 0}^{\infty} t \zeta(t) \prod_{t^{\prime} = 0}^{t - 1} (1 -
\zeta(t^{\prime}))$.
The expected value of the geometric distribution, which is the
distribution of waiting time until the first success when the
Bernoulli trial with a constant success probability $p$ is repeated,
is equal to $1/p$, i.e., the number $n$ of trials satisfying $n p =
1$.
Following this idea, we referred to the minimum $\tau \in \mathbb{N}$
satisfying $\sum_{t = 0}^{\tau} \zeta(t) \geq 1$ as the pseudo
expected number of trials repeated until the first success in the
above Poisson trial is obtained and denote it by $\E^{\ast}(t)$.
Note that there does not necessarily exist $t \in \mathbb{N}$
satisfying $\sum_{t = 0}^{\tau} \zeta(t) = 1$.
Substituting $\E^{\ast}(\tau)$ for $\E(\tau)$, we can obtain the
following result in the case where $t$ is continuous.

\begin{proposition}[Expected time before differentiation]
\label{prop:waiting_time_until_speciation_occurs}
We suppose that there exist two reproductive sequence domains $D_{1},
D_{2} \subset A^{\ast}$ under an environment $E \in \mathcal{E}$.
Let $S(t)$ be a population of DNA sequences at time $t$ from an
ancestor sequence born in $D_{1}$.
Then, the pseudo expected time $\E^{\ast}(\tau)$ until $S(t)$
differentiates into $D_{2}$ is provided by minimum $\tau \in [0,
\infty)$ satisfying $\int_{0}^{\tau} \zeta(t) \md t \geq 1$.
\end{proposition}

\begin{proof}
Obvious from Lemma \ref{lemma:prob_speciation} and the definition
of $\E^{\ast}(t)$.
\end{proof}

Equation (\ref{eq:speciation_probability}),
consequently, $\E^{\ast}(\tau)$ includes the $d$th power for the
distance $d$ from $S(t)$ to $D_{2}$.
Therefore, the time elapses before the differentiation occurs is
expected to become longer rapidly as the distance to other
reproductive sequence domain increases.
To compute $\E^{\ast}(\tau)$ actually, it is important to estimate the
selection pressure $p(s, t, E)$, mutation probability $\pi$,
population size $n(t)$, and location of a different reproductive
sequence domain $D_{2}$.
In some cases $\pi$ and $n(t)$ can be estimated and predicted, but
it may be difficult for us human beings to estimate $D_{2}$.
A method for estimating $p(s, t, E)$ was investigated in
\cite{Koyano_submitted}.

\section{Conclusion}
\label{sec:Conclusion}

In this study, we constructed the evolutionary model of a population
of DNA sequences on the noncommutative topological monoid $A^{\ast}$
formed by strings on the alphabet $A$ composed of four letters
$\mathtt{a}, \mathtt{c}, \mathtt{g}$, and $\mathtt{t}$ and analyzed
the constructed model in a theoretical manner.
No molecular-evolutionary models have been constructed on $A^{\ast}$,
however $A^{\ast}$ is a natural stage for molecular-evolutionary
modeling because DNA sequences are represented as elements of
$A^{\ast}$.
It is a task in the future to tackle problems in molecular evolution
that are difficult to address theoretically, applying the model
developed here in a numerical manner.


\newcommand{\noop}[1]{}

\end{document}